\documentclass{amsart}

\newtheorem{theorem}{Theorem}[section]
\newtheorem{lemma}[theorem]{Lemma}
\newtheorem{proposition}{Proposition}[section]
\usepackage{amsmath,amssymb,amsthm}
\theoremstyle{definition}
\newtheorem{definition}[theorem]{Definition}
\newtheorem{example}[theorem]{Example}

\newtheorem{corollary}[theorem]{Corollary}
\theoremstyle{remark}
\newtheorem{remark}[theorem]{Remark}
\usepackage{enumitem}
\numberwithin{equation}{section}
\begin{document}
	
	\title{MDS and $I$-Perfect Codes in Pomset block Metric}
	
	
	\author{Atul Kumar Shriwastva}
	\address{{Department of Mathematics, National Institute of Technology Warangal, Hanamkonda, Telangana 506004, India}}
	\email{shriwastvaatul@student.nitw.ac.in}
	\thanks{}
	
	\author{R. S. Selvaraj}
	\address{{Department of Mathematics, National Institute of Technology Warangal, Hanamkonda, Telangana 506004, India}}
	\email{rsselva@nitw.ac.in}
	\date{\today}
\begin{abstract}
	In this paper, we establish the Singleton bound for pomset block codes ($(Pm,\pi)$-codes) of length $N$ over the ring $\mathbb{Z}_m$. We give a necessary condition for a code to be MDS in the pomset (block) metric and prove that every MDS $(Pm,\pi)$-code is an MDS $(P,\pi)$-code. Then we proceed on to find $I$-perfect and $r$-perfect codes. Further, given an ideal with partial and full counts, we look into how MDS and $I$-perfect codes relate to one another. For chain pomset, we obtain the duality theorem for pomset block codes of length $N$ over  $\mathbb{Z}_m$; and, the weight distribution of MDS pomset block codes is then determined. 
\end{abstract}
\maketitle
\section{Introduction}	
  The main problem in coding theory is to find the largest minimum distance $d$ of any $k$-dimensional linear code of length $n$ over the finite field $\mathbb{F}_q$ for any integer $n > k \geq 1$. This problem for Hamming space was generalized by Niederreiter \cite{hnt}. Motivated by Niederreiter, Brualdi et al. developed poset space \cite{Bru} by using a partially ordered relation on the set $[n]$, where $[n] = \{1,2,\ldots,n\}$ represents the coordinate positions of $n$-tuples in the vector space $\mathbb{F}_q^n$. 
  Over the past two decades, the study of block codes in coding theory has sparked several significant developments in the communication field, such as experimental design, high-dimensional numerical integration, and cryptography.
 K. Feng et al. introduced block codes of length $N$ over $\mathbb{F}_q $ known as $\pi$-block codes \cite{fxh} in 2006, which is another generalization of Hamming codes.
 The author Alves et al. \cite{Ebc} extends it to $(P,\pi)$-block codes using a partial order relation on the block positions of $\mathbb{F}_q^N$. Block codes can be explored with the various metrics, allowing one to study the class of posets (such as hierarchical posets, NRT posets, crown posets, etc.).
 \par  
 A well-studied class of maximum distance separable codes was also investigated in the space $\mathbb{F}_q^N$ with metrics such as poset metric, $\pi$-metric, and poset block metric, and a weight distribution for this class was determined in \cite{hkmdspc}, \cite{fxh}, and \cite{bkdnsr}, respectively. MDS codes have applications in both combinatorics and finite geometry. In \cite{fxh}, and \cite{bkdnsr}, authors have shown that if all the blocks are not the same size, then the dual of an MDS block code need not be MDS. But in the case of chain poset we prove that dual of an MDS block code is MDS (see the Theorem \ref{Duality theorem}).
 \par With the aid of a partial order on a regular multiset, I. G. Sudha and R. S. Selvaraj proposed a new term in coding theory known as pomset codes \cite{gsrs} over the ring $\mathbb{Z} m$. The authors also established the MacWilliam type identities for linear codes \cite{gsr}, MDS, and $I$-perfect codes \cite{gr}.
  The concept of multiset theory \cite{BWD} was established by  W. D. Blizard (1989). For more information, one can see \cite{GKJ} and \cite{gsrs}.
 However, L. Panek citewcps recently (2020) proposed the weighted coordinate poset metric on $\mathbb{F}_q^n$, which is based on a weight function and partial order on $[n]$. This is a simpler version of the pomset metric that doesn't need the multiset structure. 
 \par We observed that there is a problem for researchers, what could be the $I$-balls in \cite{wcps} for an ideal $I$, whereas the $I$-balls and their properties for the pomset metric are fully described in \cite{gr} as the form of an ideal with a full count and a partial count. A fresh approach to its study are brought about by the fact that an ideal can be a full or partial count in the pomset space. It entices scholars to concentrate on it instead of \cite{BWD} and \cite{wcps}. In \cite{bkdnsr}, given an ideal $I$ in poset block space, an $[N, k] $ $ (P, \pi)$-code $\mathbb{C}$ is $I$-perfect if and only if there is a function 
 $	f : \bigoplus\limits_{j \notin I^*} \mathbb{F}_q^{k_j} \rightarrow \bigoplus\limits_{i \in I^*} \mathbb{F}_q^{k_i} $
 such that $ \mathbb{C} = \{ (v, f(v)  ) : v \in \bigoplus\limits_{j \notin I^*} \mathbb{F}_q^{k_j}\}$. If we consider an ideal $I$ with full counts, then it is true (see the Theorem \ref{Iperfect}), but this need not be true in the case of an ideal with a partial count (see the example \ref{IperfectExam}). We also prove that if a $(Pm, \pi)$-code $\mathbb{C}$ is $r$-perfect, then there will be unique ideal $I$ of the cardinality $r$.
As in \cite{fxh}, unless all the blocks have the same dimension, the dual of an MDS block code is not necessary to be MDS. However, if we assume that pomset is a chain, then the dual of an MDS block code is MDS (see the Theorem \ref{Duality theorem}).
  \par  In this paper, we introduce MDS pomset block codes of length $N$ over the ring $\mathbb{Z}_m$ and extend the concept of $I$-perfect and $r$-perfect pomset codes to the pomset block metric. Section ${2}$ establishes the basic properties of multiset, ideal, and pomset.
   Section $3$ begins by defining the $(Pm,\pi)$-metric (or pomset block metric) on the space $\mathbb{Z}_m^N$. Then, we establish the Singleton bound for pomset block codes of length $N$ over the ring $\mathbb{Z}_m$. We also compare the maximum distance separability of $(Pm,\pi)$-codes with different poset metric structures (see the Theorem \ref{MDScompare}). In particular, when all the blocks have the same length, a necessary condition for a code to be MDS in pomset (block) metric is found. Then we look for $I$-perfect codes and $r$-perfect codes (see the Theorems \ref{Iperfect} and \ref{r-perf}), as well as the link between MDS and $I$-perfect codes with partial and full counts. Section $5$ examines the $(Pm, \pi)$-codes with chain pomset and proves the duality theorem for such codes. Moreover, we determine the weight distribution of MDS pomset block codes.
  \section{Preliminaries}
 \sloppy{Given a finite set $B =\{ b_1, b_2, \ldots, b_n\}$, let $C_M : B \rightarrow \mathbb{N}$ be a counting function such that $C_M (b_i) = b_i$ $\forall$ $b_i \in B$. A Multiset (in short, mset) is a collection of elements wherein repetition is allowed. That is, a multiset $ M=\{c_1/b_1, c_2/b_2, \dots, c_n/b_n \}$ is drawn from the set $B$, where  $ c_i/i \in M $  represents an element $ b_i \in B $ appears $ c_i $ times in $M$. If $C_M(b_i)=h$ for some positive integer $h$, $M$ is called regular and $h$ is called its height. The cardinality of multiset $M$ is $|M| = \sum\limits_{j \in B} C_M (j) $ and the root set of $M$ is $ M^{*} \triangleq \{j \in B : C_M (j) > 0 \}$. 
  \par The mset space $[B]^r$ is the set of all regular multisets $M$  of height $r$ drawn from $B$. Let $M_1,M_2 \in [A]^r$. 
 If $C_{M_1} (b) \leq C_{M_2} (b) $ $\forall$ $b \in B$, then $M_1$ is called as submultiset (or submset) of $M_2$ $(M_1 \subseteq M_2)$, otherwise it is said to be proper ($M_1 \subset M_2$).
  Union:  $ M_1 \cup M_2 \triangleq \{ C_{M_1 \cup M_2} (b)/b : C_{M_1 \cup M_2} (b) =  \max\{C_{M_1} (b), C_{M_2} (b)\}  ~ \forall ~  b \in B\}$.
  Intersection:  $  M_1 \cap M_2 \triangleq \{ C_{M_1 \cap M_2} (b)/b : C_{M_1 \cap M_2} (b) =  \min \{C_{M_1} (b), C_{M_2} (b)\} ~ \forall ~ b \in B\}$. Mset sum: $  M_1 \oplus M_2 \triangleq \{ C_{ M_1 \oplus M_2} (b)/b : C_{ M_1 \oplus M_2} (b) = \min \{C_{M_1} (b)+C_{M_2} (b), r\} \text{ for all } b \in B\}$. The mset difference of $M_2$ from $M_1 $: 
  $  M_1 \ominus M_2 \triangleq \{ C_{ M_1 \ominus M_2} (b)/b : C_{ M_1 \ominus M_2} (b) = \max \{C_{M_1} (b)-C_{M_2} (b), 0\} \text{ for all } b \in B\}$. 
  Cartesian product:  $M_1 \times M_2 =  \{mn/(m/a,n/b) : m/a \in M_1 \text{ and }  n/b \in M_2\}$. The notation $t/(m/a, n/b) $ means that the pair $(a, b)$ is appearing $t$ times in $M_1 \times M_2$ where $1 \leq  t \leq mn$.  For an mset $ M \in [B]^r$, if every member $(m/a,n/b) \in R$ has count $C_{M} (a) \times C_{M} (b) $, then the submset $R$ of $M \times M $ is said to be mset relation on $M$. The compliment
 of $M$ is $ M^c =\{C_{M^c} (b)/b : C_{M^c} (b) = r - C_M(b) $ for all $b \in  B \}$.}
  \par A mset relation $R \subseteq M \times M$ with: (1) for every $p/a \in M$, $p/aR p/a$ (reflexive), (2) if $p/a R q/b$ and $q/b R p/a$ then $p=q, a=b$ (anti-symmetric), and  (3) if $p/a R q/b$ and $q/b R t/c$ then $p/a R t/c$ (transitive); is said to be partially ordered mset  relation (or pomset relation). Pomset, which is represented by the symbol $\mathbb{P}$, is the pair $(M, R)$. An element $p/a \in M$ is a maximal element of $\mathbb{P}$ if there is no any $q/c \in M$ such that $p/a R q/c$. An element $r/c \in M$ is a minimal element of $\mathbb{P}$ if there is no any $q/b \in M$ such that $q/b R r/c$. $\mathbb{P}$ is a chain iff every distinct pair of $M$ is comparable in $\mathbb{P}$. $\mathbb{P}$ is said to be an antichain if every distinct pair of elements from $M$ is not comparable in $\mathbb{P}$.
  \par A submset $I$ of $M$ is called an  ideal of $\mathbb{P}$ if $p/a \in I$ and $q/b ~R ~ p/a~(b \neq a)$ imply $q/b \in I$. An ideal generated by an element $p/a \in M$ is defined as $\langle p/a \rangle = \{p/a\} \cup \{q/b \in M : q/b R p/a \} $. An ideal generated by $I$  is defined as $\langle I \rangle = \bigcup\limits_{p/a \in I} \langle p/a \rangle $.  An ideal $I$ is said to be of full count if $C_{I}(i)=C_{M}(i)$ for every $i \in I^*$ otherwise it is ideal $I$ with a partial count. Throughout the paper, $\mathcal{I}(\mathbb{P})$ denote the set of all ideals in $\mathbb{P}$
  and $\mathcal{I}^t(\mathbb{P}) $ denote the set of all ideals of cardinality $t$ in $\mathbb{P}$.
  	Given pomset $ \mathbb{P}  = (M,R)$,  the dual pomset $\tilde{\mathbb{P} }= (M, \tilde{R})$ with the same underlying mset $M $ such that $p/a \tilde{R} q/b $ in $ \tilde{P}$ if and only if $q/b R p/a$ in  $ P $.
  	The	order ideals of $\tilde{\mathbb{P} }$ are $\mathcal{I}(\tilde{ \mathbb{P} }) = \{I^c : I \in \mathcal{I}(\mathbb{P}) \}$.
\section{Pomset Block Metric Space ($(Pm,\pi)$-space)}
In this Section, we start with the basic definition of $(Pm, \pi) $-spaces (for more details, one can see \cite{AS}). Then, we establish the Singleton bound for $(Pm,\pi)$-codes and derive MDS codes.
\par Considering the regular multiset $ M=\{\lfloor\frac{m}{2} \rfloor/1, \lfloor \frac{m}{2} \rfloor/2, \dots, \lfloor \frac{m}{2} \rfloor/n \}$ drawn from  $[n]$ with $R$ to be partial order, the pair  $\mathbb{P}=(M,R)$ is a pomset.
	Let  $\pi$ be a label map from $[n] $ to $ \mathbb{N}$  defined by $\pi(i) = k_i$ with $\sum_{i=1}^{n}\pi (i) = N$. Consider the space $ \mathbb{Z}_m^N $ over the ring of integers modulo $m$ such that  $ \mathbb{Z}_m^N $ is the direct sum of modules $ \mathbb{Z}_{m}^{k_1}, \mathbb{Z}_{m}^{k_2}, \ldots, \mathbb{Z}_{m}^{k_n} $. That is,  $\mathbb{Z}_{m}^{N} = \mathbb{Z}_{m}^{k_1}  \oplus \mathbb{Z}_{m}^{k_2} \oplus \ldots \oplus \mathbb{Z}_{m}^{k_n} $.  Every $N$-tuple  $v$ in $ \mathbb{Z}_{m}^{N} $ can be expressed uniquely as $	v = v_1 \oplus v_2 \oplus \ldots \oplus v_n $, where $v_i = (v_{i_1},v_{i_2},\ldots,v_{i_{k_i}}) \in \mathbb{Z}_{m}^{k_i}$.  
	For $v_i \in \mathbb{Z}_{m}^{k_i} $,  Lee support of $v_i  \in \mathbb{Z}_{m}^{k_i}$ is defined as $supp_{L}(v_i) =\{c_{i_j}/i_j : c_{i_j}=w_{L}(v_{i_j}),c_{i_j} \neq 0, \text{ for }  1\leq j \leq k_i\}$. Let  $ Max_c supp_{L}(v_i) \triangleq \max\{ c_{i_j} : c_{i_j}/{i_j} \in {supp_{L}(v_i)}  \} $ denotes  the maximum among the Lee weights of the components of $v_i$.
	Then, we define 
	the block support or $(Pm,\pi)$-support of  $ v  \in \mathbb{Z}_{m}^{N} $  as 
	\begin{align*}
		supp_{(Pm,\pi)}(v) = \{ c_i/i \in M : v_i \neq 0 ~\text{and} ~c_i = Max_c supp_{L}(v_i) \}
	\end{align*}
	the submultiset of $M$. 
	\par The $ (Pm,\pi) $-weight of is defined $v$ as $ w_{(Pm,\pi)}(v) = |\langle supp_{(Pm,\pi)}(v) \rangle| $.  $ (Pm,\pi) $-distance between  $ u, v \in \mathbb{Z}_{m}^N $ is given by $d_{(Pm,\pi)}(u,v) = w_{(Pm,\pi)}(u - v)$.
	$d_{(Pm,\pi)}(.,.)$ defines a metric over $ \mathbb{Z}_{m}^N $ called as pomset block metric (or $(Pm,\pi)$-metric). 
	The pair $ (\mathbb{Z}_{m}^N,~d_{(Pm,\pi)} ) $ is said to be a pomset block space. 
	\par 
	Let $\mathbb{C} \subseteq \mathbb{Z}_{m}^N $ be a pomset block code (or say $(Pm,\pi)$-code) of length $N$. The minimum distance of the $(Pm,\pi)$-code $\mathbb{C}$ is given by  
	$d_{(Pm,\pi)}\mathbb{(C)} = \min \{  d_{(Pm, \pi)} {(c_1, c_2)}: c_1, c_2 \in \mathbb{C} \}$.  If $\mathbb{C}$ is linear, then the minimum distance of  $\mathbb{C}$ becomes $\centering d_{(Pm,\pi)}\mathbb{(C)} = \min \{ w_{(Pm,\pi)}(c) : 0 \neq c \in \mathbb{C} \} $.
	As  $ w_{(Pm,\pi)}(v) \leq n \lfloor \frac{m}{2} \rfloor $ for any $v \in \mathbb{Z}_{m}^N $, the  minimum distance of  any $(Pm,\pi)$-code $\mathbb{C} $ is bounded above by $ n \lfloor \frac{m}{2} \rfloor $.
	\begin{remark}\label{r3}
		If $k_i = 1 $ $\forall$ $i \in [n]$ then the  pomset block space $ (\mathbb{Z}_{m}^N,~d_{(Pm,~\pi)}) $  becomes the classical pomset space $ (\mathbb{Z}_{m}^n,~d_{Pm})$ \cite{gsrs}. Pomset block metric $d_{(Pm,\pi)}$ extends the classical pomset metric, which accommodates	the Lee metric introduced by I. G. Sudha and  R. S. Selvaraj, in particular. It generalizes the poset block metric introduced by M. M. S Alves et al., in general, over $\mathbb{Z}_m$.	
	\end{remark}
\par In \cite{gr}, I. G. Sudha and R. S. Selvaraj established the Singleton bound for any pomset code $\mathbb{C}$ of length $n$ over $\mathbb{Z}_m$ and extend the concept of MDS codes (in terms of full and partial counts) to the pomset space where each block have length $k_i =1$ for every $i \in [n]$. In this paper, we extend the concept of MDS pomset codes and $I$-perfect pomset codes ($k_i =1$ for every $i \in [n]$) to the pomset block codes with $k_i = \pi(i)$ for every $i \in [n]$ such that $ \sum\limits_{i=1}^{n} \pi (i) = N  $. 
Let  $\mathcal{I}_{*r}^{t}$ be the collection of all ideals in $\mathcal{I}(\mathbb{P})$ such that cardinality of $I$ is $t$ and cardinality of $I^*$ is $r$. 
\begin{equation*}
	\mathcal{I}_{*r}^{t} = \{ I \in \mathcal{I}(\mathbb{P}) : \ |I|=t \text{ and} \ |I^*|=r \text{ for some } r \leq n \}
\end{equation*}
\begin{theorem}[Singleton bound for pomset block code] \label{singlB}
	Let $\mathbb{C}$  be a pomset block code of length $N = k_1 + k_2 + \ldots + k_n$, over $\mathbb{Z}_m$ with minimum distance $d_{(Pm,\pi)}(\mathbb{C})$.  Then $
	\max\limits_{J \in  \mathcal{I}_{*r}^{t}}  \sum\limits_{i \in J^*} k_{i} \leq N - \lceil log_{m}|\mathbb{C}| \rceil$ where $r = \big\lfloor \frac{d_{(Pm,\pi)} \mathbb{(C)}-1} { \big\lfloor \frac{ m }{2} \big\rfloor } \big\rfloor $ and $t \leq d_{(Pm,\pi)}\mathbb{C}-1$.
\end{theorem}
\begin{proof}
	\sloppy{There exist two distinct codewords $c_1, c_2 \in \mathbb{Z}_m^N $ such that $d_{(Pm,\pi)} (\mathbb{C}) = d_{(Pm, \pi)} {(c_1,c_2)} = |\langle supp_{(Pm,\pi)}(c_1-c_2) \rangle| $. Let $I =\langle supp_{(Pm,\pi)}(c_1-c_2) \rangle $. Then, $d_{(Pm, \pi)}(\mathbb{C}) - 1 <  { \big\lfloor \frac{ m }{2} \big\rfloor }|I^*| $ so that  $\big\lfloor \frac{d_{(Pm,\pi)} \mathbb{(C)}-1}{ \big\lfloor \frac{ m }{2} \big\rfloor }\big\rfloor < |I^*|$. Let $r = \big\lfloor \frac{d_{(Pm,\pi)} \mathbb{(C)}-1}{ \big\lfloor \frac{ m }{2} \big\rfloor }\big\rfloor$ and $t = r{ \big\lfloor \frac{ m }{2} \big\rfloor }$ then $t \leq d_{(Pm, \pi)}(\mathbb{C}) -1$. Since $I$ is an ideal of cardinality $d_{(Pm, \pi)}(\mathbb{C})$, (Ref. \cite{gsrs}, by Proposition $3$), there always exist an ideal $J$ of cardinality $t$ such that $J \subseteq I$ and $|J^*|=r$.  For any ideal $J \in \mathcal{I}_{*r}^{t}$, every two distinct codewords of $\mathbb{C}$ must be different  outside of $J^*$ for some $i$-labels in $[n]$. Otherwise $ d_{(Pm, \pi)} {(c_1,c_2)} \leq |J| \leq d_{(Pm, \pi)}(\mathbb{C}) - 1$, a contradiction. Then, $|\mathbb{C}| \leq | \mathbb{Z}_m^{N - \sum\limits_{i \in J^*} k_{i} }|$. Hence, $ { \sum\limits_{i \in J^*} k_{i} } \leq  N - \lceil \log_m |\mathbb{C}| \rceil$. As this is true for any ideal $J \in \mathcal{I}_{*r}^{t}$, therefore, 
	$	\max\limits_{J \in  \mathcal{I}_{*r}^{t}} \sum\limits_{i \in J^*} k_{i} \leq N - \lceil log_{m}|\mathbb{C}| \rceil $
	where $r = \big\lfloor \frac{d_{(Pm,\pi)} \mathbb{(C)}-1}{\big\lfloor\frac{ m }{2}\big \rfloor} \big\rfloor $ and $t \leq d_{(Pm,\pi)}-1 $.}
\end{proof}
\begin{corollary} 
	Let $\mathbb{C}$  be a pomset block code of length $N$ over $\mathbb{Z}_m$ with minimum distance $d_{(Pm,\pi)}(\mathbb{C})$. Then following hold: 
	\begin{enumerate}[label=(\roman*)]	 
	\item If $\pi(i) = s$ for all $i \in [n]$ then $\big\lfloor \frac{d_{(Pm,\pi)} \mathbb{(C)}-1}{\big\lfloor\frac{ m }{2}\big \rfloor} \big\rfloor \leq n- {\frac{\lceil log_{m}|\mathbb{C}|\rceil }{s}}$. 
	\item  If $k_1\geq k_2 \geq \ldots \geq k_n$ then $  k_{n} \big\lfloor \frac{d_{(Pm,\pi)} \mathbb{(C)}-1}{\big\lfloor\frac{ m }{2}\big \rfloor} \big\rfloor  \leq N - \lceil log_{m}|\mathbb{C}| \rceil$.
	\item  If $\pi(i) = 1$ for all $i \in [n]$ then $d_{(Pm,\pi)} (\mathbb{C}) = d_{Pm}(\mathbb{C})$ and Singleton bound of pomset block code becomes Singleton bound for pomset code i.e. $	\big\lfloor \frac{d_{Pm} \mathbb{(C)}-1}{\big\lfloor\frac{ m }{2}\big \rfloor} \big\rfloor  \leq n - \lceil log_{m}|\mathbb{C}| \rceil$.
	\end{enumerate}
\end{corollary}
\begin{definition}
	A pomset block code $\mathbb{C} \subseteq \mathbb{Z}_m^N $ of length $N$ over $\mathbb{Z}_m$ is said to be a maximum distance separable pomset block code if it attains its Singleton bound, that is, there exists an  ideal $J \in  \mathcal{I}_{*r}^{t}$ for a pomset block code $\mathbb{C}$  such that 
	$	\sum\limits_{i \in J^*} k_{i} = N - \lceil log_{m}|\mathbb{C}| \rceil$.
\end{definition}
For the case $k_i = s $  $\forall$ $i \in [n]$, we give a necessary condition for a code to be MDS with respect to pomset (block) metric in the following successive results. 
\begin{theorem}
	Let $k_i = s $  $\forall$ $i \in [n]$ and $\mathbb{C} $  be a $(Pm,\pi)$-code of length $N$  over  $\mathbb{Z}_m$ with minimum distance $d_{(Pm,\pi)}(\mathbb{C})$.  If $\mathbb{C}$ is MDS, then $
	{\big\lfloor\frac{ m }{2}\big \rfloor}(n - \frac{\lceil log_{m}|\mathbb{C}| \rceil}{s})    + 1  \leq  d_{(Pm,\pi)} \mathbb{(C)} \leq {\big\lfloor\frac{ m }{2}\big \rfloor}(n - \frac{\lceil log_{m}|\mathbb{C}| \rceil}{s}+1) $. 
\end{theorem}
\begin{proof}
	Since $\mathbb{C}$ is  MDS and $\pi(i) = s$ $\forall$ $i$, then $\big\lfloor \frac{d_{(Pm,\pi)} \mathbb{(C)}-1}{\big\lfloor\frac{ m }{2}\big \rfloor} \big\rfloor = n- {\frac{\lceil log_{m}|\mathbb{C}|\rceil }{s}}$. Thus, we have 	$	n - \frac{\lceil log_{m}|\mathbb{C}| \rceil}{s} \leq  \frac{d_{(Pm,\pi)} \mathbb{(C)}-1}{\big\lfloor\frac{ m }{2}\big \rfloor} < n - \frac{\lceil log_{m}|\mathbb{C}| \rceil}{s} +1 $. Hence, $
	{\big\lfloor\frac{ m }{2}\big \rfloor}(n - \frac{\lceil log_{m}|\mathbb{C}| \rceil}{s})  + 1 \leq  d_{(Pm,\pi)} \mathbb{(C)} \leq {\big\lfloor\frac{ m }{2}\big \rfloor}(n - \frac{\lceil log_{m}|\mathbb{C}| \rceil}{s}+1)$.		
\end{proof}
Thus, if $k_i = s $ $\forall$ $i$ and  $\mathbb{C} $ is a $(Pm,\pi)$-code of length $N$ over $\mathbb{Z}_m$ with minimum distance $d_{(Pm,\pi)}(\mathbb{C})$, then
$\mathbb{C}$ cannot be an  MDS  whenever $	1 \leq d_{(Pm,\pi)} \mathbb{(C)} \leq 	{\big\lfloor\frac{ m }{2}\big \rfloor}(n - \frac{\lceil log_{m}|\mathbb{C}| \rceil}{s})  $ or $ d_{(Pm,\pi)} (\mathbb{C})  > {\big\lfloor\frac{ m }{2}\big \rfloor}(n - \frac{\lceil log_{m}|\mathbb{C}| \rceil}{s}+1) $.
\begin{theorem}
	If $\mathbb{C} $ is an MDS pomset code of length $n$ over $\mathbb{Z}_m$ with minimum distance $d_{Pm}(\mathbb{C})$, then $
	{\big\lfloor\frac{ m }{2}\big \rfloor}(n - \lceil log_{m}|\mathbb{C}| \rceil )  + 1  \leq  d_{Pm} \mathbb{(C)} \leq {\big\lfloor\frac{ m }{2}\big \rfloor}(n - \lceil log_{m}|\mathbb{C}| \rceil+1) $. 
\end{theorem}
 Thus, if $\mathbb{C} $ is a pomset code of length $n$ over $\mathbb{Z}_m$ with minimum distance $d_{Pm}(\mathbb{C})$, then $\mathbb{C}$ cannot be an MDS whenever $1 \leq d_{(Pm,\pi)} \mathbb{(C)} \leq 	{\big\lfloor\frac{ m }{2}\big \rfloor}(n - \lceil log_{m}|\mathbb{C}| \rceil )  $ or $ d_{(Pm,\pi)} (\mathbb{C})  > {\big\lfloor\frac{ m }{2}\big \rfloor}(n - \lceil log_{m}|\mathbb{C}| \rceil +1) $.
\par Now we will compare the maximum distance separability of $(Pm,\pi)$-codes with different poset metric structures. Consider $P=([n], \preceq_P)$ to be a poset induced by the pomset $\mathbb{P}$ such that if $a/i R b/j$ in pomset $\mathbb{P}$  implies $i \preceq_P j$ in the poset $P$. 
\begin{proposition} \label{phg}
	Let $d_{(Pm,\pi)}(\mathbb{C})$  and  $d_{(P,\pi)}(\mathbb{C})$ be the minimum distances of a code $\mathbb{C} $ of length $N$ over $\mathbb{Z}_m$ with respect to $(Pm,\pi)$-metric and $(P,\pi)$-metric respectively.  Then $	\big\lfloor \frac{d_{(Pm,\pi)} \mathbb{(C)}-1}{ \big\lfloor\frac{ m }{2}\big \rfloor}\big\rfloor \leq d_{(P,\pi)}(\mathbb{C}) - 1$.	
\end{proposition}
\begin{proof}
	\sloppy{There exist two distinct  codewords  $c_1, c_2 $ in $\mathbb{C} $ such that $d_{(P,\pi)} (\mathbb{C}) = d_{(P, \pi)} {(c_1,c_2)} = |\langle supp_{\pi}(c_1-c_2) \rangle|$. Assume that $ d_{(Pm,\pi)} \mathbb{(C)} > \big\lfloor\frac{ m }{2}\big \rfloor d_{(P,\pi)}\mathbb{(C)}$. Then $ w_{(Pm, \pi)} {(c_1-c_2)} \leq \lfloor\frac{ m }{2} \rfloor |\langle supp_{\pi}(c_1-c_2) \rangle| < d_{(Pm,\pi)} \mathbb{(C)}$, a contradiction. Thus,  $ d_{(Pm,\pi)} \mathbb{(C)} \leq \big\lfloor\frac{ m }{2}\big \rfloor d_{(P,\pi)}\mathbb{(C)}$, $d_{(Pm, \pi)}(\mathbb{C}) - 1 <  \big\lfloor\frac{ m }{2}\big \rfloor d_{(P,\pi)} \mathbb{(C)} $. Hence, $\big\lfloor \frac{d_{(Pm,\pi)} \mathbb{(C)}-1}{ \big\lfloor \frac{ m }{2} \big\rfloor }\big\rfloor \leq d_{(P,\pi)} \mathbb{(C)}-1$.} 
\end{proof}
\begin{proposition} \label{phg}
	Let $d_{Pm}(\mathbb{C})$  and  $d_{P}(\mathbb{C})$ be the minimum distances of a code $\mathbb{C} $ of length $n$ over $\mathbb{Z}_m$ with respect to pomset metric and poset metric respectively.  Then $	\big\lfloor \frac{d_{Pm} \mathbb{(C)}-1}{ \big\lfloor\frac{ m }{2}\big \rfloor}\big\rfloor \leq d_{P}(\mathbb{C}) - 1$.	
\end{proposition}
Now, we will determine whether an MDS code with respect to $(Pm,\pi)$-metric is also an  MDS code with respect to $(P,\pi)$-metric or not in the following successive results. Recall that, Singleton bound \cite{bkdnsr} of any $(P,\pi)$-code $\mathbb{C} $  is $ \max\limits_{J \in  \mathcal{I}^{d_{(P,\pi)}(\mathbb{C}) - 1}}  \big\{\sum_{i \in J} k_{i}\big\} \leq  N - \lceil log_{m}|\mathbb{C}| \rceil $. 
\begin{theorem}\label{MDScompare}
	Every MDS $(Pm,\pi)$-code is an MDS $(P,\pi)$-code.
\end{theorem}
	\begin{proof}
		\sloppy{Let $\mathbb{C} $ be an  MDS $(Pm,\pi)$-code. Then $ \max\limits_{J \in  \mathcal{I}_{*r}^{t}}  \big\{\sum_{i \in J^*} k_{i}\big\} = N - \lceil log_{m}|\mathbb{C}| \rceil $ where $r = \big\lfloor \frac{d_{(Pm,\pi)} \mathbb{(C)}-1}{M_w} \big\rfloor$. As $ r \leq d_{(P,\pi)}(\mathbb{C}) - 1 $ by  Proposition \ref{phg}, we have 
			$ \max\limits_{J \in  \mathcal{I}_{*r}^{t}}  \big\{\sum_{i \in J^*} k_{i}\big\} \leq \max\limits_{J \in  \mathcal{I}^{d_{(P,\pi)}(\mathbb{C}) - 1}}  \big\{\sum_{i \in J} k_{i}\big\}$.  	Thus,  $  N - \lceil log_{m}|\mathbb{C}| \rceil \leq  \max\limits_{J \in  \mathcal{I}^{d_{(P,\pi)}(\mathbb{C}) - 1}}  \big\{  \sum_{i \in J} k_{i}\big\}$.
			Hence $\mathbb{C} $  is MDS  with respect to $(P,\pi)$-metric.}
	\end{proof}
\begin{theorem}\label{MDScompare}
	Every MDS pomset code is an MDS poset code.
\end{theorem}
\par  
The pomset block ball (or $(Pm,\pi)$-ball) with center $ u \in \mathbb{Z}_{m}^N $ and radius $r$ is defined as $B_{(Pm,\pi)}(u, r) = \{ v \in \mathbb{Z}_{m}^N : d_{(Pm,\pi)}(u, v) \leq r \}$. For a mset $I$ in $\mathbb{P}$, the $I$-ball centered at $u$ is $B_{I}(u) \triangleq \{ v \in \mathbb{Z}_{m}^N :  supp_{(Pm,\pi)}(u- v)  \subseteq I \}$. For $v \in  B_I (u)$, it is not necessary that $\langle supp_{(Pm,\pi)} {(u - v)} \rangle \subseteq  I$. If $I $ is an ideal in $P $, then $\langle supp_{(Pm,\pi)} (u - v) \rangle \subseteq I$ is always true. Hence, for an ideal $I$ in $\mathbb{P}$, the $I$-ball centered at $u$ is $B_{I}(u) \triangleq \{ v \in \mathbb{Z}_{m}^N : \langle supp_{(Pm,\pi)}(u- v) \rangle \subseteq I \}$. For each $ v \in B_I(u)$, $d_{(Pm,\pi)}(u,v) \leq |I|$. The $I$-sphere centered at $u$ is $S_{I}(u) \triangleq \{ v \in \mathbb{Z}_{m}^N : \langle supp_{(Pm,\pi)}(u-v) \rangle = I \}$. Let $ \mathcal{I}_j^i $ denote the collection of all ideals with cardinality $ i $ having $ j $ maximal elements.  Let $Max(I) = \{c_{i_1}/i_1,c_{i_2}/i_2,\ldots,c_{i_j}/i_j \}$ denote the set of maximal elements in the ideal $I$. Then 
$\bigcup\limits_{j=1}^{\min\{i,n\}} \mathcal{I}_j^i = \mathcal{I}^i{(\mathbb{P})}$.
\par Unlike the results in poset space \cite{hkmdspc}, $I$-balls in \cite{gr} are no longer linear when the ideal $I$ has a partial count. In a similar way, we noticed that $I$-balls behave the same way in the $(Pm,\pi)$-metric as does in the pomset-metric \cite{gr}. $I$-balls remain linear when $I$ is an ideal with a full count, and its properties are similar to the ideal in poset space. The following Proposition is a generalization of [ref. \cite{gr}, Proposition $3$]  and because the proof follows the same pattern as \cite{gsrs} it was omitted. 
\begin{proposition} \label{full count theorem}
	\sloppy{Let $\mathbb{P} $ be a pomset on a regular mset $ M=\{\lfloor\frac{m}{2} \rfloor/1, \lfloor \frac{m}{2} \rfloor/2, \dots, \lfloor \frac{m}{2} \rfloor/n \}$ and $\tilde{\mathbb{P} }$ be the dual pomset of $\mathbb{P} $. If $I$ is an ideal with a full count in $\mathbb{P}$, then} 
	\begin{enumerate}[label=(\roman*)]
		\item $B_I$ is a submodule of $\mathbb{Z}_m^N$ and $|B_{I}| = m ^{\sum\limits_{i \in I^*} k_i}$.
		\item  For $u \in \mathbb{Z}_m^N$, $B_{I}(u)$ is the coset of $B_I$ containing $u$, ie. $B_I(u) = u + B_I$. 
		\item For $u,v \in \mathbb{Z}_m^N$, the two $I$-balls $B_I(u) $ and $B_I(v)$ are either identical or disjoint. Moreover,
		$	B_I(u)  =  B_I(v)  \text{ if and only if  } supp_{(Pm, \pi)}{(u-v)} \subseteq I$.
	\end{enumerate} 
\end{proposition}
Hence,	$(\mathbb{Z}_m^N,~d_{(Pm,\pi)})$-space can be partitioned into $I$-balls for  an ideal  $ I $ with full count.
\begin{remark}
	Given an ideal $I$ with partial count, $I$-ball need not be a submodule, but for $u \in \mathbb{Z}_m^N$, $B_{I}(u)$ is the translate of $B_I$, ie. $B_I(u) = u + B_I$.
\end{remark}
\begin{proposition}\label{runnion}
	Every $r$-ball is a union of all $I$-balls where $I$ is an ideal of cardinality $r$ i.e. $	B_{(Pm,\pi)}(u,r)= \bigcup\limits_{I \in  \mathcal{I}^r(P)}B_{I}(u)$.
\end{proposition}
\begin{definition}
	 A $(Pm, \pi)$-code $\mathbb{C}$ of length $N$ over $\mathbb{Z}_m$ is said to be $I$-perfect if the  $I$-balls centered at the codewords of $\mathbb{C}$ are pairwise disjoint and their union is $\mathbb{Z}_m^N$.
	 $ \mathbb{C}$ is an $r$-error correcting $(Pm, \pi)$-code if the $(Pm, \pi)$-balls of radius $r$ centered at the codewords of $ \mathbb{C}$ are pairwise disjoint.
	It is said to be $r$-perfect if the $r$-balls centered at the codewords of $C$ are pairwise disjoint and their union covers the entire space $\mathbb{Z}_{m} ^ N$. 
\end{definition}
For an ideal $I$ with a full count, from Proposition \ref{full count theorem},  the space $\mathbb{Z}_m^N$ can be partitioned into $I$-balls. The number of $I$-balls is  $ m^{N -\sum\limits_{i \in I^*}k_i}$. Then the set of  collection of exactly one tuple from each $I$-balls forms an $I$-perfect $(Pm,\pi)$-code $\mathbb{C}$ of length $N$ over  $\mathbb{Z}_m$ and  $ |\mathbb{C} | = m^{N -\sum\limits_{i \in I^*}k_i}$ with  $d_{(Pm,\pi)} (\mathbb{C}) > |I|$. On the other hand, for an ideal $I$ with the partial count, $I$-balls  need not behold the properties given in the above Theorems.
\par The following Lemma is true only in the case when $I$ is an ideal with full counts. It need not be true in the case of an ideal with partial count.
\begin{lemma}\label{|B int C|=1}
	Let $I$ be an ideal with full count and $\mathbb{C}$ be an $[N, k] $ $ (Pm, \pi)$-code. Then the following are equivalent:
	\begin{enumerate}[label=(\roman*)]
		\item $\mathbb{C}$ is an $I$-perfect code.
		\item $	\sum\limits_{i \in I^*} k_{i} = N-k$ (the covering condition)
		and  $ | \mathbb{C} \cap B_I| =1 $ (the packing condition).
		\item $ | \mathbb{C} \cap B_I (y)| =1 $ $\forall$ $y \in \mathbb{Z}_m^N$; that is each element of  $ \mathbb{Z}_m^N$ belongs to precisely one $I$-ball centered at a codewords of  $\mathbb{C}$.
	\end{enumerate}  
\end{lemma}
The following Theorem demonstrates that $I$-perfect codes are systematic, with information symbols in $I^{*c}$ blocks and parity check symbols in $I^*$ blocks. Thus, $I$-perfect codes are easier to deal with than $r$-perfect codes. 
 We write $x \in \mathbb{Z}_m^N$ as $(x_1, x_2) $ where $x_1 \in \bigoplus\limits_{j \notin I^*} \mathbb{Z}_m^{k_j} $ and $x_2 \in \bigoplus\limits_{i \in I^*} \mathbb{Z}_m^{k_i} $.
\begin{theorem}\label{Iperfect}
	Given an ideal $I$ with full count, an $[N, k] $ $ (Pm, \pi)$-code is $I$-perfect if and only if there is a function 
	$	f : \bigoplus\limits_{j \notin I^*} \mathbb{Z}_m^{k_j} \rightarrow \bigoplus\limits_{i \in I^*} \mathbb{Z}_m^{k_i} $
	such that $ \mathbb{C} = \{ (v, f(v)  ) : v \in \bigoplus\limits_{j \notin I^*} \mathbb{Z}_m^{k_j}\}$.
\end{theorem}
\begin{proof}
	Assume that the code $ \mathbb{C}$ is an $I$-perfect. Let $v \in \bigoplus\limits_{i \notin I^*} \mathbb{Z}_m^{k_i}$. Since $ \mathbb{C}$ is an $I$-perfect code,  $(v, 0) \in B_I (c)$ for some $c \in \mathbb{C}$ which implies that $supp_{(Pm,\pi)}(c - (v, 0)) \subseteq I$ and $ c - (v, 0)=(0, u) \in B_I$. This proves that $c =(v, u)$. Moreover, if there is another element $c'=(v, w) \in C$, then $c - c' =(0, u - w)$. As $I$ is an ideal with full count, then $supp_{(Pm,\pi)}(0, u - w) \subseteq I$ so $ c - c' \in B_I$ which implies that $c \in B_I(c')$, a contradiction. Therefore, we can define a function $
	f : \bigoplus\limits_{j \notin I^*} \mathbb{Z}_m^{k_j} \rightarrow \bigoplus\limits_{i \in I^*} \mathbb{Z}_m^{k_i} 
	$	which sends  $v \in  \bigoplus\limits_{j \notin I^*} \mathbb{Z}_m^{k_j}  $ to the unique $ u \in \bigoplus\limits_{i \in I^*}\mathbb{Z}_m^{k_i}  $ such that $c=(v,u)$.
	\par Conversely, assume that such a function exists. Then, for any $(v, u) \in  \mathbb{Z}_m^N$, we have $ \mathbb{C} \cap B_I (v, u) ={(v, f(v))} $. By Lemma \ref{|B int C|=1}, we get that $\mathbb{C} $ is an $I$-perfect code.
\end{proof}
	Since $ \mathbb{C}$ is a linear subspace of $\mathbb{Z}_m^N$, such type of a function $f : \bigoplus\limits_{j \notin I^*} \mathbb{Z}_m^{k_j} \rightarrow \bigoplus\limits_{i \in I^*} \mathbb{Z}_m^{k_i}$ is a linear transformation. 
\par We shall see that the example below, which assumes an ideal with a partial count, demonstrates that the aforementioned theorem need not apply.
\begin{example}\label{IperfectExam}
	{Consider the space $\mathbb{Z}_{10}^3$ with $k_i = 1 $  $\forall$ $i \in [3]$. Consider $\mathbb{C} = \{(0,0,0), (1,0,0), (2,0,0), (3,0,0), (4,0,0),(5,0,0), (6,0,0), (7,0,0), (8,0,0), (9,0,0),\\(0,5,0), (1,5,0), (2,5,0), (3,5,0), (4,5,0), (5,5,0), (6,5,0), (7,5,0), (8,5,0), (9,5,0) \}$. Let $\mathbb{P}$ be an antichain pomset defined on the regular multiset $M= \{5/1,5/2,5/3\}$.  Let $I = \{2/2,5/3\}$ be an ideal with partial count.}  
	For $ (a,0,0) \in \mathbb{C}$ where $a \in \mathbb{Z}_{10}$, the $I$-ball centered at $B_{I}(a,0,0) = \{ (a, b,c) : b \in\{ 0,\pm1, \pm 2 \} ~ \textit{and} ~ c \in \mathbb{Z}_{10}\}$ and  $|B_{I}(a,0,0)| = 50$. For $ (a,5,0) \in \mathbb{C}$ where $a \in \mathbb{Z}_{10}$, the $I$-ball centered at $B_{I}(a,5,0) = \{ (a, b,c) : b \in\{ \pm 3, \pm 4,5 \} ~ \textit{and} ~ c \in \mathbb{Z}_{10}\}$ and  $|B_{I}(a,5,0)| = 50$. Clearly, all balls are distinct and disjoint. Since $ |\mathbb{C}| |B_I| =20 \times 50 = 1000= |\mathbb{Z}_{10}^3| $. Hence, $  \mathbb{C}$ is an $I$-perfect. Take $x = (1,0,0)$, $y = (1,5,0)$  $ \in \mathbb{C}$. Here $v =(1) \in \bigoplus\limits_{j \notin I^*}  \mathbb{Z}_{10}^{k_j}$ in $x$ and $y$ but  image of $v$ is $(0,0)$ and $(5,0)$. 
\end{example}
 We derive a necessary and sufficient condition for a pomset block code $\mathbb{C}$ to be an $r$-error correcting code in terms of $I$-balls in the following Theorem. 
\begin{theorem}
	Let $\mathbb{C}$ be a $(Pm, \pi)$-code of length $N$ over $\mathbb{Z}_m$. Then $\mathbb{C}$ is an $r$-error correcting code if and only if for any two distinct codewords $u,v \in \mathbb{C}$, $u-v \notin B_{I \cup J}$ $\forall$ $I, J \in \mathcal{I}^r$.
\end{theorem}
\begin{proof}
	Assume that $\mathbb{C}$ is an $r$-error correcting code. Let $u,v \in \mathbb{C}$ such that $u\neq v$. Suppose that $u-v \in B_{I \cup J}$ for some  $I, J \in \mathcal{I}^r$. Choose $ x \in \mathbb{Z}_m^N$ such that $ x= u - (u-v)_{I^* \setminus J^*} $ where
	$(u-v)_{I^*\setminus J^*} $ means $u_j-v_j=0$ if $j \in J^*$.
	Then, $d_{(Pm,\pi)}(x,u) = | \langle supp_{(Pm,\pi)} (u-x)| = | \langle supp_{(Pm,\pi)} ( u-u+(u-v)_{I^* \setminus J^*} )| = | \langle supp_{(Pm,\pi)} ((u-v)_{I^* \setminus J^*} )| \leq  | \langle I \rangle|=r $. Thus, $x \in B_{(Pm,\pi)}(u,r)$. Now, $d_{(Pm,\pi)}(x,v) = | \langle supp_{(Pm,\pi)} (v-x)| = | \langle supp_{(Pm,\pi)} ( (u-v)- (u-v)_{I^* \setminus J^*} )| = | \langle supp_{(Pm,\pi)} ((u-v)_{J^* \setminus I^*} )| \leq  | \langle I \rangle|=r $. So that, $x \in B_{(Pm,\pi)}(v,r)$ also. It is a contradiction that $\mathbb{C}$ is an $r$-error correcting code.
	\par Conversely, assume that  $\mathbb{C}$ is not an $r$-error correcting code. There exist  two distinct codewords $u,v \in \mathbb{C}$ and $y \in \mathbb{F}_q^N$  such that 
	$y \in  B_{(Pm,\pi)}(u,r) \cap B_{(Pm,\pi)}(v,r)$. That is $| \langle supp_{(Pm,\pi)} (u-y) \rangle| \leq r$ and  $| \langle supp_{(Pm,\pi)} (v-y) \rangle| \leq r$. Let 
	$ \langle supp_{(Pm,\pi)} (u-y) \rangle =I$  and  $ \langle supp_{(Pm,\pi)} (u-y) \rangle =J$  where $|I| \leq r $ and $|J| \leq r$. (Ref. \cite{gsrs},  Proposition $4$), there exist an  ideal $I' \in \mathcal{I}^r $, $ J'\in \mathcal{I}^r $  such that $I \subseteq I'$ and $J \subseteq J'$. Then, $  supp_{(Pm,\pi)} (u-v)  = supp_{(Pm,\pi)} (u-y+y-v) \subseteq supp_{(Pm,\pi)} (u-y) \cup  supp_{(Pm,\pi)} (v-y) \subseteq I' \cup J' $. Hence, 
	$u-v \in B_{I'\cup J'}$ for some $I', J' \in \mathcal{I}^r $, contradiction.
\end{proof}
In particular, for $I=J $, the above theorem implies that for an $r$-error correcting pomset block code, the $I$-balls centered at the codewords of $\mathbb{C}$ are disjoint for each $I \in \mathcal{I}^r$.
\par Now, we will determine the connection between $r$-perfectness and $I$-perfectness of a pomset (block) codes in the following successive results.
 \begin{theorem}\label{r-perf}
 	Let  $\mathcal{I}^r(\mathbb{P}) =\{I\}$ for some $r \leq n \lfloor \frac{ m }{2} \rfloor$. Then a  $(Pm, \pi)$-code $\mathbb{C}$ is $r$-perfect if and only if $\mathbb{C}$ is $I$-perfect.
 \end{theorem}
 \begin{proof}
 Let $ \mathbb{C} $ be an $r$-perfect $(Pm, \pi)$-code, 
 	$ \mathbb{Z}_m^N = \bigcup\limits_{c \in  \mathbb{C} }^{\circ}B_{(Pm,\pi)}(c,r) $.  From Proposition \ref{runnion}, we have
 	$B_{(Pm,\pi)}(c,r)= \bigcup\limits_{I \in \mathcal{I}^r(P)}B_{I}(c)$. As  $\mathcal{I}^r(\mathbb{P}) =\{I\}$, we have $ \mathbb{Z}_m^N = \bigcup\limits_{c \in  \mathbb{C} }^{\circ}B_{I}(c) $. Hence, $\mathbb{C}$ is an $I$-perfect $(Pm, \pi)$-code.
 \end{proof}
 \begin{theorem}
 	A $(Pm, \pi)$-code $ \mathbb{C} $ of cardinality $m^k$ over  $ \mathbb{Z}_{m}$ is $(N-k)$-perfect if and only if  $\mathbb{C}$ is an $I$-perfect $(Pm, \pi)$-code and $ \mathcal{I}^{N-k}  =\{I\}$.
 \end{theorem}
 \begin{proof}
 	Let $ \mathbb{C} $ be $(N-k)$-perfect.  Suppose that $\mathbb{C}$ is not $I$-perfect for some $I \in \mathcal{I}^{N-k}$, then there exist two distinct codewords  $b, c \in \mathbb{C} $ such that 
 	$ \langle supp_{(Pm,\pi)}(b-c) \rangle \subseteq I$. Thus, $ |\langle supp_{(Pm,\pi)}(b-c) \rangle | \leq N-k$ and so $c \in B_{(Pm,\pi)}{(b, N-k)}$ which is condradiction of $ \mathbb{C} $ is $(N-k)$-perfect. 
 	Now, suppose that $\mathbb{C}$ is $I$-perfect and  $\{I\} \subsetneq \mathcal{I}^{N-k}$ which imply  $ |B_{(Pm,\pi)}{(b, N-k)}| > |B_{I}{(b)}| = m^{N-k}$. Again, $(Pm, \pi)$-code $ \mathbb{C} $ is not $(N-k)$-perfect, a contradiction. The converse follows straight forword from Theorem \ref{r-perf}.
 \end{proof}
 \begin{remark}
 	If $ |\mathcal{I}^{N-k}| \geq 2 $ then there does not exixt any $(N-k)$-perfect $(Pm, \pi)$-code  $ \mathbb{C} $ of cardinality $m^k$ over  $ \mathbb{Z}_{m}$.
 \end{remark}
 \begin{theorem}\label{rperfect in pomset}
	Let $ \mathbb{C} $ be a pomset code of length $n$ over $\mathbb{Z}_{m}$ of cardinality $m^k$. Then $ \mathbb{C} $ is $(n-k)$-perfect if and only if  $\mathbb{C}$ is an $I$-perfect pomset code and $ \mathcal{I}^{n-k}  =\{I\}$.
\end{theorem}
Given an ideal with partial and full counts in the next successive results, we continue to look at the relationship between MDS and $I$-perfect codes.
\begin{theorem}  If $\mathbb{C}$ is an MDS block code of length $N$ over $\mathbb{Z}_m$ with cardinality $m^k$ for some $k > 0$  then $\mathbb{C}$ is an $I$-perfect for all ideals $I \in \mathcal{I}^{r \lfloor \frac{m}{2} \rfloor }$ with full count. 
\end{theorem}
\begin{proof}
	Let $I \in \mathcal{I}^{r \lfloor \frac{m}{2} \rfloor }$ be an ideal with full count. From Proposition \ref{full count theorem},  $\mathbb{Z}_m^N$ can be partitioned into $I$-balls. Let $l$ be the number of $I$-balls, so we have $l|B_I| = m^N$ and $l = | \mathbb{C} | $. Since $\mathbb{C} $ is  MDS then there exist an ideal $J \in I_{*r}^{t}$ such that $\sum\limits_{i \in J^*}k_i = N-k$ and $| J^*| = \big\lfloor \frac{d_{(Pm,\pi)} \mathbb{(C)}-1}{ \big\lfloor \frac{ m }{2} \big\rfloor }\big\rfloor $. If $d_{(Pm,\pi)}(\mathbb{C}) \leq |J|$ then $d_{(Pm,\pi)}(\mathbb{C}) \leq \big\lfloor\frac{ m }{2}\big \rfloor  |J^*|$, $\frac{d_{(Pm,\pi)} \mathbb{(C)}-1}{ \big\lfloor \frac{ m }{2} \big\rfloor } < |J^*|$, not possible. Thus, $d_{(Pm,\pi)}(\mathbb{C}) > |I|$.  Since $I$ is an ideal with full count, then again from Proposition \ref{full count theorem}, any two $I$-balls centered at distinct codewords of $\mathbb{C}$ must be disjoint and  $|\mathbb{C}||B_I| = m^{k} m^{\sum\limits_{i \in I^*}k_i} = m^{N}$.  Hence, $\mathbb{C}$ is $I$-perfect.
\end{proof}
\begin{example}
	{Let $M=\{2/1, 2/2,2/3\}$ be a regular multiset on $\{1,2,3\}$ and  $R$ be a pomset relation defined on $M$ such that $R= \{4/(2/1, 2/1), 4/(2/2,  2/2),4/(2/3, 2/3), 4/(2/1,  2/2), 4/(2/1,  2/3)\}$.
	Consider the space  $\mathbb{Z}_5^{7}= \mathbb{Z}_5^2 \oplus \mathbb{Z}_5^4 \oplus  \mathbb{Z}_5^1$ with the $(Pm, \pi)$-metric where $k_1=2,k_2=4,$ and $k_3=1$.
	Let $\mathbb{C} = \{(0,0,0,0,0,0,0),(0,3,0,2,0,0,1), (0,1,0,4,0,0,2),  (0,4,0,1,0,0,3), (0,\\2,0,3,0,0,4)\}$ be a linear code. Here $|\mathbb{C}|=5$ and $w_{(Pm,\pi)}(0,3,0,2,0,0,1)=w_{(Pm,\pi)}(0,1,0,4,0,0,2)=w_{(Pm,\pi)}(0,4,0,1,0,0,3)=w_{(Pm,\pi)}(0,2,0,3,0,0,4)=5$}. Thus, $d_{(Pm,\pi)}(\mathbb{C})=5$. We have $r= \big\lfloor \frac{d_{(Pm,\pi)} \mathbb{(C)}-1}{ \big\lfloor \frac{ m }{2} \big\rfloor }\big\rfloor=2$ and $\mathcal{I}_{*r}^t=\{ \{2/1, 2/2\}, \{2/1, 2/3\},\{2/1,1/2\},\{2/1, 1/3\}\}$. Clearly, $ \max\limits_{J \in  \mathcal{I}_{*r}^{t}}  \big\{\sum_{i \in J^*} k_{i}\big\} =6=N-k$, $\mathbb{C}$ is MDS. In $\mathbb{P}$, $I_1=\{2/1, 2/2\}$ and $ I_2=\{2/1, 2/3\}$ are ideals with full count, and $I_3=\{2/1,1/2\}$, $I_4=\{2/1, 1/3\}$ are ideals with partial count. One can see that $\mathbb{C}$ is $I_1$-perfect as well as $I_2$-perfect, but neither $I_3$-perfect nor $I_4$-perfect.
\end{example}
\begin{theorem} \label{Iperfect imply MDS}
	Let $\mathbb{C} $ be a pomset block code of length $N$ over $\mathbb{Z}_m$. If $\mathbb{C}$ is $I$-perfect for all ideals $I \in \mathcal{I}^{\lfloor \frac{m}{2} \rfloor (N- \lceil \log_m|\mathbb{C}| \rceil)}$ then $\mathbb{C}$ is MDS.
\end{theorem}
\begin{proof}
	Let $\mathbb{C}$ be $I$-perfect for all ideals $I \in \mathcal{I}^{\lfloor \frac{m}{2} \rfloor (N-\lceil \log_m|\mathbb{C}| \rceil )}$. Suppose that $x$ and $y$ are two distinct codewords in $\mathbb{C}$ such that $ d_{(Pm,\pi)}(x,y) \leq  \big\lfloor \frac{m}{2} \big\rfloor (N-\lceil \log_m|\mathbb{C}| \rceil ) $. Let $ \langle supp_{(Pm,\pi)}{(x-y)} \rangle = J$ then $x-y \in B_J $. (Ref. \cite{gsrs}, by Proposition $3$), there exist an ideal $I$ of cardinality ${\lfloor \frac{m}{2} \rfloor }(N-\lceil \log_m|\mathbb{C}| \rceil )$  containing $J$ such that  $J \subseteq I$. Thus, $x-y \in B_I $ and $x \in B_I(y) $. So, $\mathbb{C}$ would not be $I$-perfect, a contradiction. We have $ d_{(Pm,\pi)}(\mathbb{C}) >  \big\lfloor \frac{m}{2} \big\rfloor (N-\lceil \log_m|\mathbb{C}| \rceil ) $, then  
	$ \frac{d_{(Pm,\pi)}(\mathbb{C}) - 1}{\big\lfloor \frac{m}{2} \big\rfloor } \geq  N-\lceil \log_m|\mathbb{C}| \rceil $. Thus,  $ \big\lfloor \frac{d_{(Pm,\pi)}(\mathbb{C}) - 1}{\big\lfloor \frac{m}{2} \big\rfloor } \big\rfloor \geq  N-\lceil \log_m|\mathbb{C}| \rceil  $. Then, we have 
		$	\max\limits_{J \in  \mathcal{I}_{*r}^{t}} \sum_{j \in J^*} k_{j} \geq \big\lfloor \frac{d_{(Pm,\pi)}(\mathbb{C}) - 1}{\big\lfloor \frac{m}{2} \big\rfloor } \big\rfloor \geq N-\lceil \log_m|\mathbb{C}| \rceil $.
	Hence, $\mathbb{C}$ is MDS.
\end{proof}
\begin{theorem} \label{Iperfect imply MDS same length}
	Let $\mathbb{C} $ be a pomset block code of length $N$ over $\mathbb{Z}_m$ with cardianlity $m^k$. If $\mathbb{C}$ is $\lfloor \frac{m}{2} \rfloor (N-k)$-perfect then $\mathbb{C}$ is MDS.
\end{theorem}
\begin{theorem} \label{Iperfect imply MDS pomset}
	Let $\mathbb{C} $ be a pomset code of length $n$ over $\mathbb{Z}_m$ with cardianlity $m^k$. If $\mathbb{C}$ is $\lfloor \frac{m}{2} \rfloor (n-k)$-perfect then $\mathbb{C}$ is MDS.
\end{theorem}
\section{Block codes in chain pomsets}
Throughout this Section, $\mathbb{P}=(M,R)$ is considered as  chain. So that for $1 \leq j \leq n$, $|\mathcal{I}_j^{t}| =1$ and each ideal $I$ in $\mathbb{P}$ has a unique
maximal element. Let $ c_{i_t} / i_t $ be the maximal element of $I$ and the remaining elements of $I^*$ have full count.
Let $v \in \mathbb{Z}_{m}^N$. Then,
 we have
\begin{equation*}
	w_{(Pm,\pi)}(v)= c_{i_t} + (| \langle {c_{i_t}/i_t} \rangle ^*|-1) \lfloor\frac{ m }{2} \rfloor
\end{equation*} 
The space $ (\mathbb{Z}_{q}^N, ~d_{(Pm,\pi)} )$ is called the  NRT pomset block space (when ${\mathbb{P}}$ is a chain).
Since $|\mathcal{I}_1^t|=1$, then we have  $B_I (x) = B_t(x)$ for any $x \in \mathbb{Z}_m^N$ and $\max\limits_{J \in  \mathcal{I}_{*r}^{t}}  \big\{\sum_{i \in J^*} k_{i}\big\} = \sum\limits_{i \in J^*} k_{i}$. From Theorem \ref{singlB}, we have Singleton bound for chain pomset block code:
\begin{theorem} \label{chain Singleton bound}
	Let $\mathbb{C} $  be a chain pomset block code of length $N = k_1 + k_2 + \ldots + k_n$, over $\mathbb{Z}_m$ with minimum distance $d_{(Pm,\pi)}(\mathbb{C})$. Then
	$ \sum\limits_{i \in I^*} k_{i} \leq N - \lceil log_{m}|\mathbb{C}| \rceil$ where $|I| \leq d_{(Pm,\pi)}(\mathbb{C})-1$ and $|I^*| = \big\lfloor \frac{d_{(Pm,\pi)} \mathbb{(C)}-1} { \big\lfloor \frac{ m }{2} \big\rfloor } \big\rfloor $.
\end{theorem}
\begin{proposition}
	For the chain Pomset $\mathbb{P}$, the following statements hold:
	\begin{enumerate}[label=(\roman*)]	 
		\item  Let $I$ be an ideal in $\mathbb{P}$ with partial count then $|B_I| = (1 + 2c_{i_t})^{k_{{i_t}}} m^{\sum\limits_{j \in I^* \setminus \{i_t\}} k_j}$ where $ c_{i_t}=|I|-(|I^*|-1)\lfloor\frac{m}{2} \rfloor $. Moreover, if $k_i =s$ $\forall$ $i \in [n]$, then  $|B_I| = (1 + 2c_{i_t})^{s} m^{ (|I^*|-1) s}$.
		\item Every $I$-perfect block code is an $r$-perfect block code and vice versa, where $r=|I|$. 
	\end{enumerate}
\end{proposition}
\begin{corollary}
	For the chain (usual) Pomset $\mathbb{P}$, An $(N, k) $ $ (Pm, \pi)$-code is $r$-perfect if and only if there is a linear transformation 
	$	L : \bigoplus\limits_{j = r+1}^{n} \mathbb{Z}_m^{k_{j}} \rightarrow \bigoplus\limits_{i = 1}^{r} \mathbb{Z}_m^{k_i} $
	such that $ \mathbb{C} = \{ (v, L(v)  ) : v \in  \mathbb{Z}_m^{k_{r+1}} \oplus \mathbb{Z}_m^{k_{r+2}} \oplus \ldots \oplus\mathbb{Z}_m^{k_{n}} \}$.
\end{corollary}
\begin{proof}
	The proof follows from Theorem \ref{Iperfect}.
\end{proof}
\begin{proposition}
	Every $I$-perfect block code with $|I|= \lfloor \frac{m}{2} \rfloor ( N - \lceil \log_m|\mathbb{C}| \rceil \big\rfloor)$ is an MDS block code.
\end{proposition}
\begin{proposition}\label{p11}
	Every $I$-perfect block code with $|I^*|= \big\lfloor \frac{d_{(Pm,\pi)} \mathbb{(C)}-1} { \big\lfloor \frac{ m }{2} \big\rfloor } \big\rfloor$ is an MDS block code.
\end{proposition}
\begin{proof}
	Let $\mathbb{C}$ be an $I$-perfect block code. Then $d_{(Pm,\pi)}(\mathbb{C}) > |I|$. If  $I$ is an ideal with full count then $|I|=  \lfloor \frac{m}{2} \rfloor |I^*|$  and  $\sum_{i \in I^*}k_i = N - \lceil \log_m|\mathbb{C}| \rceil$. Hence, $\mathbb{C}$ is an MDS block code. Now, if $I$ is an ideal  with partial count, $|I| = |J| + c_{i_t} $ for some ideal $J$ with full count and $0 < c_{i_t} \leq \lfloor \frac{m}{2} \rfloor -1$. Since $\mathbb{C}$ is an $I$-perfect code, $ |\mathbb{C}||B_I| = m^{N} $ implies
	$|\mathbb{C}|(2c_{i_t} + 1)^{k_{i_t}}|B_J| = m^{N}$. Taking $log_m$ both sides, we get,
	$	\sum\limits_{j \in J^*}k_{j} = N - \log_m|\mathbb{C}| - k_{i_t} \log_m(2c_{i_t}+1)$.
	As $0 < c_{i_t} \leq \lfloor \frac{m}{2} \rfloor -1 $, then $ 0 < \log_m (2c_{i_t}+1) <1$. Since $\log_m|\mathbb{C}| + k_{i_t} \log_m(2c_{i_t}+1)$ is an integer. Then, 
	\begin{align*}
			\log_m|\mathbb{C}| + k_{i_t} \log_m(2c_{i_t}+1) &= \lfloor \log_m|\mathbb{C}| + k_{i_t} \log_m(2c_{i_t}+1) \rfloor \\
			& \leq  \lfloor \log_m|\mathbb{C}| \rfloor + \lfloor k_{i_t} \log_m(2c_{i_t}+1) \rfloor +1 \\
		&	<  \lfloor \log_m|\mathbb{C}| \rfloor +  k_{i_t} \lceil \log_m(2c_{i_t}+1) \rceil +1 \\
		&<  \lfloor \log_m|\mathbb{C}| \rfloor +  k_{i_t} \\
	&	\leq  \lceil \log_m|\mathbb{C}| \rceil +  k_{i_t}-1
	\end{align*}
 Thus, $	\sum\limits_{j \in J^*}k_{j} = N - (\log_m|\mathbb{C}| +  k_{i_t} \log_m(2c_{i_t}+1) ) \geq 
	N - \lceil \log_m|\mathbb{C}| \rceil -  k_{i_t} +1 $.  We get,
	$\sum_{j \in I^*}k_{i} = \sum_{j \in J^*}k_{j} + k_{i_t} \geq N - \lceil \log_m|\mathbb{C}|$. Hence, $\mathbb{C}$ is an MDS- code. 
\end{proof}
\begin{theorem}\label{t14}
	Let  $\mathbb{C} \subseteq \mathbb{Z}_m^N $ be a block code of length $N$ with cardinality $m^k$. Then $\mathbb{C}$ is an MDS-block code if and only if $\mathbb{C} $ is $I$-perfect for all $I \in \mathcal{I}^{\lfloor \frac{m}{2} \rfloor (N - k)}$.	
\end{theorem}
We have (ref. \cite{AS}, Theorem $8$), if $\mathbb{C} \subseteq \mathbb{Z}_m^N $ is a linear block code of cardinality $m^k$ and $I$ be an ideal with full count in $\mathbb{P}$. Then, $\mathbb{C}$ is $I$-perfect in $\mathbb{P}$ if and only if $\mathbb{C^\perp}$ is $I^c$-perfect where $I^c$ is an ideal in $\tilde{\mathbb{P}}$. 
\begin{theorem}[Duality theorem] \label{Duality theorem}
	Let  $\tilde{\mathbb{P}}$ be the dual Pomset of the chain $\mathbb{P}$ on $M$. Let $\mathbb{C} $ be a linear pomset block code of length $N$ with cardinality $m^k$ over $\mathbb{Z}_m$, then $\mathbb{C}$ is an MDS $\mathbb{P}$-block code iff $\mathbb{C^\perp} $ is an  MDS $\tilde{\mathbb{P}}$-block code. 
\end{theorem}
\subsection{Weight distribution of MDS chain pomset block codes}
Let $\mathbb{C} $ be a pomset block code of length $N$ with cardinality $m^k$ over $\mathbb{Z}_m$ where $\mathbb{P}$ is a chain pomset. Let $A_i (\mathbb{C})= \{x \in \mathbb{C}  : w_{(Pm,\pi)}(x) = i \}$ be the weight distribution of $\mathbb{Z}_m^N$,  where $0 \leq i \leq n \lfloor \frac{m}{2} \rfloor $.
\begin{proposition} \label{mdsball}
	Let  $\mathbb{C} $ be an MDS pomset block code of length $N$ with cardinality $m^k$. For an ideal $I$ of $\mathbb{P}$,
	\begin{enumerate}[label=(\roman*)]
		\item
		If $|I| \leq \lfloor \frac{m}{2}\rfloor (N-k)$, then $B_I \cap \mathbb{C} = 1$.
		\item  If $|I| > \lfloor \frac{m}{2}\rfloor (N-k)$, then
		 \begin{align*}
			B_I \cap \mathbb{C} =
			\begin{cases}
				m^{\sum\limits_{i \in I^*} k_i - N + k}  &\text{if} ~|I| = \lfloor \frac{m}{2}\rfloor |I^*| \\
			(1 + 2c_{i_t} )^{k_{i_t}}  m^{\sum\limits_{i \in I^*} k_i - N + k - k_{i_t}} &  \text{otherwise}.
			\end{cases} 
		\end{align*}
	\end{enumerate}
 where $ c_{i_t} / i_t $ is the maximal element of an ideal $I$.
\end{proposition}
\begin{proof}
	\textit{(i)} Suppose that $|I| \leq \lfloor \frac{m}{2}\rfloor (N-k)$. Since $\mathbb{C} $ is MDS,  there exist a $J \in \mathcal{I}_{^*r}^t$ such that $\sum\limits_{i \in {J^*} } k_{i} = N-k$. So we have  $|d_{(Pm,\pi)}\mathbb{(C)}| > |J|$. Thus, only the zero vector is inside the $B_I $ and $ \mathbb{C}$.
	\par \textit{(ii)}	Suppose that $|I| > \lfloor \frac{m}{2}\rfloor (N-k)$. Since $\mathbb{C} $ is MDS,  there exist a $J \in \mathcal{I}_{^*r}^t$ such that $\sum\limits_{i \in {J^*} } k_{i} = N-k$. As $\mathbb{P}$ is a chain,  $J \subseteq I$. 
	(a) If $I$ is an ideal with a full count, then $B_I$ is a submodule of $\mathbb{Z}_m^N$, and $B_J$ is a submodule of $B_I$. The number of cosets of $B_J$ in $B_I$ is $m^{\sum\limits_{i \in {I^*} } k_i -\sum\limits_{j \in {J^*}} k_j}$. Since $\mathbb{C}$ is $J$ perfect, every coset of $B_J$ in $\mathbb{Z}_m^N$  contains exactly one codeword of  $\mathbb{C}$. Thus, $B_I$ contains $m^{\sum\limits_{i \in I^*} k_i - N + k}$. 
		(b) 
	If $I$ is an ideal with the partial count. Let $ c_{i_t} / i_t $ be the maximal element of $I$. To find $B_I \cap \mathbb{C}$, we need disjoint translates of $B_J$ in $B_I$ whose union covers $B_I$. Let $K = I \ominus J$, and $K$ is not an ideal of $\mathbb{P}$ but a submset of $M$. Then cardinality of $B_K$ is $	(1 + 2c_{i_t})^{k_{i_t}}  m^{\sum\limits_{i \in I^*} k_i - N + k - k_{i_t}}$. The translates $x+B_J$, $x \in B_K$, are disjoint and their union covers $B_I$. Hence, $ B_I \cap \mathbb{C}  = 	(1 + 2c_{i_t} )^{k_{i_t}}  m^{\sum\limits_{i \in I^*} k_i - N + k - k_{i_t}}$.
\end{proof}
\begin{theorem}
	\label{18}
	Let $\mathbb{P}$ be a chain pomset with usual order ($\leq$) and $\mathbb{C}  $ be a MDS-linear block code of length $N$ of cardinality $m^k$ with  the minimum distance $d_{ (Pm,\pi)}(\mathbb{C})$. Then, $A_i (\mathbb{C}) = $
	\begin{align*}
		\begin{cases}
			1 &  \text{if} ~i = 0,\\
			0 & \text{if} ~1 \leq i \leq d_{ (Pm,\pi)}(\mathbb{C}) - 1 \\
			( m^{k_t} - (2\lfloor \frac{m}{2}\rfloor - 1)^{k_{t}}) m^{\sum\limits_{i=1}^{t-1} {k_i - N + k} } & \text{if} ~i \geq d_{ (Pm,\pi)}(\mathbb{C}) ~\text{and} ~i = \lfloor \frac{m}{2} \rfloor t \\
			(3^{k_t} - 1) m^{\sum\limits_{i =1}^{t} k_i - N + k} & \text{if} ~i \geq d_{ (Pm,\pi)}(\mathbb{C}) ~\text{and} ~i = \lfloor \frac{m}{2} \rfloor t + 1 \\
			((2 j+1)^{k_{|I^*|}} - (2 j-1)^{k_{|J^*|}})m^{\sum\limits_{i=1}^{t} {k_i - N + k} } & \text{if} ~i \geq d_{ (Pm,\pi)}(\mathbb{C}) ~\text{and} ~i = \lfloor \frac{m}{2} \rfloor t + j  
		\end{cases} 
	\end{align*}
 where $1 < j < \lfloor \frac{m}{2} \rfloor$.
\end{theorem}
\begin{proof}
	Clearly,
	$	A_i (\mathbb{C})= \{1~ \text{if} ~i = 0; ~ \text{and} ~0 ~\text{if} ~1 \leq i \leq d_{ (Pm,\pi)}(\mathbb{C}) - 1 \}$. 
	As $\mathbb{P}$ is a chain, then
	$ A_i (\mathbb{C}) = |S_I \cap \mathbb{C}| = |B_I \cap \mathbb{C}| - |B_J \cap \mathbb{C}|$, where $|I| = i$ and $|J | = i - 1$. If $i= t \lfloor \frac{m}{2} \rfloor$, then the ideal $I$ has full count, $J$ has partial count with $|I^*| = |J^*|$.  From proposition \ref{mdsball},
	\begin{align*}
		A_i (\mathbb{C}) &= |B_I \cap \mathbb{C}| - |B_J \cap \mathbb{C}|  \\ &= m^{\sum\limits_{i =1}^{|I^*|} k_i - N + k} - (1 + 2(|J| - (|J^*| - 1)\lfloor \frac{m}{2} \rfloor))^{k_{|J^*|}} m^{\sum\limits_{i=1}^{|J^*|-1} {k_i - N + k} } \\
		&= ( m^{k_t} - (1 + 2(i - 1 - (\frac{i}{\lfloor \frac{m}{2}\rfloor} - 1)\lfloor \frac{m}{2}\rfloor)^{k_{t}}) m^{\sum\limits_{i=1}^{t-1} {k_i - N + k} } \\
		&=  (m^{k_t} - (2\lfloor \frac{m}{2}\rfloor - 1)^{k_{t}}) m^{\sum\limits_{i=1}^{t-1} {k_i - N + k} }
	\end{align*}
	If $i= t \lfloor \frac{m}{2} \rfloor + 1$ then, ideal $I$ has partial count, $J$ has full count and $|I^*| = |J^*| + 1$. From proposition \ref{mdsball},
	\begin{align*}
		A_i (\mathbb{C}) &= |B_I \cap \mathbb{C}| - |B_J \cap \mathbb{C}| \\ &= (1 + 2(|I| - (|I^*| - 1)\lfloor \frac{m}{2} \rfloor))^{k_{|I^*|}} m^{\sum\limits_{i=1}^{|I^*|-1} {k_i - N + k} } - m^{\sum\limits_{i =1}^{|J^*|} k_i - N + k} 
		\\
		&=  (3^{k_{t}} - 1) m^{\sum\limits_{i=1}^{t} {k_i - N + k} }
	\end{align*}
	Similarly, if $i= t \lfloor \frac{m}{2} \rfloor + j$ where $1 < j < \lfloor \frac{m}{2} \rfloor$, then the ideal $I$ has	partial count and $J$ also has partial count with $|I^*| = |J^*|$. From proposition \ref{mdsball},
	\begin{align*}
		A_i (\mathbb{C}) &= (1 + 2(|I| - (|I^*| - 1)\lfloor \frac{m}{2} \rfloor))^{k_{|I^*|}} m^{\sum\limits_{i=1}^{|I^*|-1} {k_i - N + k} } - \\& \hspace{4.5cm}(1 + 2(|J| - (|J^*| - 1) \lfloor \frac{m}{2} \rfloor))^{k_{|J^*|}}  m^{\sum\limits_{i=1}^{|J^*|-1} {k_i - N + k} }
		\\
		&= ( (1 + 2(|I| - (|I^*| - 1)\lfloor \frac{m}{2} \rfloor))^{k_{|I^*|}} -\\& \hspace{4.4cm} (1 + 2(|J| - (|J^*| - 1)\lfloor \frac{m}{2} \rfloor))^{k_{|J^*|}})  m^{\sum\limits_{i=1}^{|J^*|-1} {k_i - N + k} } 
	\end{align*} 
Hence, $ A_i (\mathbb{C}) = ((2 j+1)^{k_{|I^*|}} - (2 j-1)^{k_{|J^*|}})m^{\sum\limits_{i=1}^{t} {k_i - N + k} }$.
\end{proof}
\bibliographystyle{amsplain}
		
\end{document}